\documentclass[letterpaper, 10pt, conference]{ieeeconf}  %

\usepackage{balance} %
\usepackage{url}
\usepackage{subcaption}
\usepackage{xcolor}
\usepackage{cite}
\usepackage{algorithmicx}
\usepackage{algorithm}
\usepackage[noend]{algpseudocode}

\newtheorem{theorem}{Theorem}
\newtheorem{definition}{Definition}

\newcommand{\revzero}[1]{{\color{black} #1}}
\newcommand{\revone}[1]{{\color{black} #1}}
\newcommand{\revral}[1]{{\color{black} #1}}
\newcommand{\revfin}[1]{{\color{black} #1}}

\IEEEoverridecommandlockouts                              %

\usepackage{graphics} %
\usepackage{epsfig} %
\usepackage{times} %
\usepackage{amsmath} %
\usepackage{amssymb}  %

\title{\Large\bf
Distributed Online Planning for Min-Max Problems \\ in Networked Markov Games
}

\author{Alexandros E. Tzikas$^{1}$, Jinkyoo Park$^{2}$, Mykel J. Kochenderfer$^{1}$, and Ross E. Allen$^{3}$%
\thanks{$^{1}$A. E. Tzikas (corresponding author) and M. J. Kochenderfer are with the Department of Aeronautics and Astronautics, Stanford University, Stanford, CA 94305, U.S.A.
        {\tt\small \{alextzik, mykel\}@stanford.edu}}%
\thanks{$^{2}$J. Park is with the Department of Industrial and Systems Engineering, KAIST,
Yuseong-gu, Daejeon 305-701, Republic of Korea.
{\tt\small jinkyoo.park@kaist.ac.kr}}%
\thanks{$^{3}$R. E. Allen is with the MIT Lincoln Laboratory, Lexington, MA 02421-6426, U.S.A.
{\tt\small ross.allen@ll.mit.edu}}%
\thanks{Accepted to appear in the IEEE Robotics and Automation Letters.}
\thanks{© 2024 IEEE.  Personal use of this material is permitted.  Permission from IEEE must be obtained for all other uses, in any current or future media, including reprinting/republishing this material for advertising or promotional purposes, creating new collective works, for resale or redistribution to servers or lists, or reuse of any copyrighted component of this work in other works.}
}

\begin{document}

\maketitle
\thispagestyle{empty}
\pagestyle{empty}

\begin{abstract}
    Min-max problems 
    are 
    \revone{important} in multi-agent sequential decision-making because they improve the performance of the worst-performing agent in the network.
    However, solving the \revfin{multi-agent} min-max problem is challenging. We propose a modular, distributed, online planning-based algorithm that is able to approximate the solution of the min-max objective in networked Markov \revone{games}, assuming that the agents communicate within a network topology \revzero{and the transition and reward functions are neighborhood-dependent}. This \revral{set-up} is encountered in the multi-robot setting. Our method consists of two phases at every planning step. In the first phase, each agent obtains sample returns based on its local reward function, by performing online planning. 
    Using the samples from online planning, each agent constructs a concave approximation of its underlying  local return as a function of only the action 
    \revral{of its neighborhood}
    at the next planning step. In the second phase, the agents deploy a distributed optimization framework that converges to the optimal immediate next action for each agent, based on the function approximations 
    of the first phase. 
    We demonstrate our algorithm's performance through formation control \revone{simulations.}
\end{abstract}

\section{Introduction}

\PARstart{M}{ulti}-agent reinforcement learning (MARL) involves multiple independent agents that operate in a common environment, each \revfin{aiming} to optimize a long-term reward by interacting with the environment and the other agents \cite{zhang2021multi}. It has recently witnessed widespread success in many problems, such as the game of Go \cite{Go}, autonomous driving \cite{driving}, and simulated soccer \cite{soccer}.

MARL techniques can be applied to a large class of multi-agent decision problems. In this work, we will focus on the class of Markov games.
The goal of each agent in the Markov game is to determine an optimal policy 
that \revral{optimizes} its expected cumulative reward.

The Markov game can be extended to the distributed paradigm. In this setting, we assume that the agents communicate within a network topology. The agents can exchange information only with their neighbors, as defined by a communication graph.
This paradigm is critical in applications where there is no central controller communicating with all agents. The absence of a central controller \revral{offers various benefits}: i) the system does not have a single point of failure, ii) \revral{the distributed computation allows for deployment when the agents have limited communication capabilities} or the number of agents is large. Usually, the goal of the agents in this setting is to maximize the expected cumulative team-average reward,
by determining their optimal policies under the imposed communication restrictions \cite{zhang2021multi}. 

Although other learning objectives have been explored in the literature for Markov games in the distributed setting \cite{zhang2021multi}, to the best of the authors' knowledge, limited prior work \cite{zhang2014fairness} has investigated the problem where the network of agents aims to jointly maximize the worst-performing agent's expected cumulative reward (termed the max-min or min-max problem from here onward)\revral{ \cite[pg. 19 \& 24]{zhang2021multi}.}
This is an important problem when fairness is required in the system. 
\revral{Most prior works consider a shared objective among the agents \cite{lauer2000algorithm}, a team-average reward \cite{kar2012qd}, and a zero-sum set-up of agent objectives \cite{littman1994markov}. }

Various notions of fairness have been proposed in the literature, but in this work we focus on the notion of improving the performance of the worst-performing agent, which is known as the egalitarian objective \cite{zimmer2021learning}, and has \revral{applicability in multi-robot systems \cite{kishimoto2008optimized}}.
Our contributions are as follows. 
\begin{itemize}
    \item We propose a distributed algorithm \revral{for the sequential decision-making problem of a multi-agent system with the egalitarian objective. Our algorithm consists of an online planning and a distributed optimization module, and solves the max-min Markov game problem under assumptions. We are able to combine online planning and distributed optimization by creating convex function approximations of the expected cumulative costs of the agents at each timestep}. We assume the reward of each agent is dependent on the actions of its neighborhood and the portion of the state relevant to its neighborhood. We also assume that the transition probability for the portion of the state relevant to its neighborhood is only dependent on its neighborhood quantities. 
    \item  We demonstrate the performance of our algorithm \revfin{on} the formation control application. \revral{We pose the problem in min-max form and solve it through disciplined convex programming \cite{grant2006disciplined}.} Our algorithm outperforms other baselines in this problem and its performance is comparable to the optimal action sequence.
\end{itemize}

This paper is organized as follows. In section \revral{II}, we describe the prior related work. In section \revral{III}, we include the problem statement, while in section \revral{IV} we provide a qualitative description of our method. In section \revral{V} we describe our algorithm and in section \revral{VI} we test its performance in simulation. We conclude the paper in section \revral{VII}.

\section{Related Work}
\subsection{Solving the Multi-Agent Min-Max Problem}
Prior work has attempted to address the min-max multi-agent sequential decision-making problem, but there are limitations to the existing methods. 
A decentralized deep learning-based model to extract policies in this setting has been proposed \cite{zimmer2021learning}. The learned policy of every agent is only a function of the agent's current observation, which means that inter-agent communication is not utilized in the decision-making step to allow for more informed decisions. The algorithm presented in this prior work also uses a minibatch of data at every timestep, which requires the simulation of the whole system. This is not always possible in a networked multi-agent system, because the agents can only communicate with their neighbors. In addition, the authors studied the problem with the assumption that the reward of each agent is not dependent on the actions of other agents: at each agent an advantage function, which is only a function of the agent's quantities, is used. A variant of the problem has also been studied in the centralized setting, where the objective is to minimize a weighted sum of the worst agent cumulative cost and the cumulative cost of all agents \cite{zhang2014fairness}. This approach is limited: it requires a centralized controller and only solves a problem related to the min-max problem.

In the case of static distributed optimization, there exist algorithms that solve the static \revfin{multi-agent} max-min problem \cite{distr_opt}. It has also been shown that the one-step min-max problem is equivalent, under specific conditions, to a min-sum optimization problem \cite{distr_opt}. Such problems are solvable using many distributed optimization algorithms \cite{boyd2011distributed, shorinwa2023distributed1, shorinwa2023distributed2}.

\subsection{Online Planning Methods}

In single-agent Markov decision problems, online methods are often used \cite{kochenderfer2022algorithms}
when the optimal policy of the underlying problem is computationally expensive (or even intractable) to determine offline. Online methods choose the agent's next action based on reasoning about states reachable from the current state and the future cumulative rewards that can potentially be received. Many variants of online planning exist in the literature. An algorithm called partially observable Monte Carlo planning with observation widening
(POMCPOW) can handle continuous action and observation spaces by using weighted particle filtering in the Monte Carlo tree search (MCTS) \cite{sunberg2018online}. 

Solving the Markov game analytically is difficult, even as a centralized problem. A major challenge is the computational complexity. 
Online planning has been explored in the context of centralized MARL with a cooperative objective. The idea of factored-value MCTS has been combined with the max-plus algorithm to obtain an anytime online planning algorithm that is computationally efficient \revfin{because} it takes advantage of the factored structure of the value function \cite{choudhury2021scalable}. 

When we want to solve the Markov game in a distributed manner, the problem \revfin{difficulty increases}. An iterative, distributed solution method is then required that learns the optimal policies as the agents continue to interact with the environment. \textit{By integrating MCTS into our multi-agent algorithm for the max-min problem, we obtain an anytime \revfin{method} that can provide action suggestions at the current timestep, without the need to solve the entire problem. }

Decentralized online planning algorithms have been proposed in the context of the Markov game \cite{best2019dec, czechowski2020decentralized}. Czechowski and Oliehoek \cite{czechowski2020decentralized} introduce a decentralized online planning approach 
for the Markov game in the cooperative case, where each agent shares a common reward. 
Their approach is based on MCTS and allows only one agent to improve its policy at a time, using learned models for the other agents' policies. 
The problem we are considering differs from this application in three significant ways: our problem objective is not equivalent to the collaborative objective explored in this prior work, we do not require cyclic communication among the agents for the action (policy) selection, and we do not require simulations with the current joint policy in order to model the agents' behavior. In addition, a decentralized planning method based on MCTS has been proposed for scenarios where each agent makes sequential decisions in order to optimize a global objective that depends on the action sequences of all agents \cite{best2019dec}. This algorithm cannot solve our problem of interest because it assumes that the global objective is known to all agents. In the max-min problem we are considering, each agent only knows its local reward function.

\section{Problem Statement}
In our formulation, we consider the framework of the Markov game, which is defined below.
\begin{definition}\label{MG}
A Markov game is defined by the tuple 
\begin{equation}
    \left( \mathcal{N}, \mathcal{S}, \lbrace \mathcal{A}^i \rbrace_{i \in \mathcal{N}}, P, 
\lbrace R^{i}\rbrace_{i \in \mathcal{N}}, \gamma \right),
\end{equation} where 
$\mathcal{N} = \lbrace 1, \dots, N \rbrace$ 
is the set agents, 
$\mathcal{S}$ is the state space observed by all agents, $\mathcal{A}^i$ is the action space of agent$~i$. Let $\mathcal{A} = \times_{i \in \mathcal{N}} \mathcal{A}^i$ be the set of joint agent actions, \revral{where $\times$ denotes the Cartesian product}. Then $P: \mathcal{S} \times \mathcal{A} \rightarrow \Delta(\mathcal{S})$ defines \revral{the transition probabilities by \revfin{mapping} each state-joint action pair to a distribution on the state space, which belongs to the set $\Delta (\mathcal{S})$}. The reward function $R^i : \mathcal{S} \times \mathcal{A} \times \mathcal{S} \rightarrow \mathbb{R}$  determines the immediate reward received by agent$~i$ for a transition from $(\mathbf{s}, \mathbf{a})$ to $\mathbf{s}'$. The discount factor is $\gamma \in [ 0, 1)$.
\end{definition}

Our objective is to maximize the worst-performing agent's expected cumulative reward, under constrained communication at every timestep. The constrained communication is imposed by an undirected graph $\mathcal{G} = \left( \mathcal{N}, \mathcal{E} \right)$, where $\mathcal{E}$ is the set of communication links, i.e., $(i, j) \in \mathcal{E}$ if and only if node $i$ can communicate directly with node $j$. We assume that at every timestep agent $i$ can send information to agent $j$ if and only if agent $j$ can send information to agent $i$, i.e. $(i, j) \in \mathcal{E} \iff (j, i) \in \mathcal{E}$. The neighborhood of agent $i$ is denoted $\mathcal{N}_i = \lbrace j \mid (i, j) \in \mathcal{E}\rbrace \cup \lbrace i \rbrace$. The environment and the interaction of the agents with the environment are expressed through the tuple \revral{$\left( \mathcal{N}, \mathcal{S}, \lbrace \mathcal{A}^i \rbrace_{i \in \mathcal{N}}, P, 
\lbrace R^{i}\rbrace_{i \in \mathcal{N}}, \gamma \right)$}. Out of the entities in the tuple, we assume that \revral{$\mathcal{A},\ P, $ and $\gamma$} are known to all agents, but $R^i$ is only known to agent $i$. We further impose the restriction:
\begin{equation} \label{eq:reward}
    R^i(\mathbf{s}, \mathbf{a}, \mathbf{s}') = R^i(\mathbf{s}^{\mathcal{N}_i}, \mathbf{a}^{\mathcal{N}_i}, \mathbf{s'}^{\mathcal{N}_i}),\ \forall i \in \mathcal{N}.
\end{equation}
\revzero{Eq. (\ref{eq:reward}) states that the reward of agent $i$ is only a function of the actions of the agents in the neighborhood of agent $i$, $\mathbf{a}^{\mathcal{N}_i} \in \mathcal{A}^{\mathcal{N}_i}=\times_{i \in \mathcal{N}_i} \mathcal{A}^i$, and of the portion of the state which is relevant to the neighborhood of agent $i$, $\mathbf{s}^{\mathcal{N}_i}$ and $\mathbf{s'}^{\mathcal{N}_i}$}. $\mathbf{a}^{\mathcal{N}_i}$ and $\mathbf{s}^{\mathcal{N}_i}$ are subvectors of the vectors $\mathbf{a}$ and $\mathbf{s}$, respectively. Eq. \eqref{eq:reward} is an extension to the common assumption that the reward of agent $i$ is dependent only on the actions of its neighbors \cite{yi2021learning}. We note that nevertheless, every agent has access to the complete current state $\mathbf{s}$, because we do not deal with partial observability in this work.

\revzero{We also assume that
\begin{equation}\label{eq:P_struct}
    P(\mathbf{s'}^{\mathcal{N}_i} \mid \mathbf{s}, \mathbf{a}) =  P(\mathbf{s'}^{\mathcal{N}_i} \mid \mathbf{s}^{\mathcal{N}_i}, \mathbf{a}^{\mathcal{N}_i}).
\end{equation}
The portion of the state relevant to the neighborhood of agent $i$ is conditionally dependent only on the current timestep's quantities of the neighborhood. This can be encountered in multi-robot systems where the state is the concatenation of the state of each robot and the robots have independent dynamics.
}

We finally introduce a new entity for every agent $i$, $I_t^{\mathcal{N}_i} \in \mathcal{I}_t^{\mathcal{N}_i}$, where $\mathcal{I}_t^{\mathcal{N}_i}$ is the set of possible pieces of information available  to agent $i$ from its neighborhood at timestep $t$. \revral{It could for example include the action sequence of agent $i$'s neighbors up to the previous timestep or incorporate knowledge about the neighbors' neighborhoods from previous timesteps, which are propagated to agent $i$ through the network}. Note that $I_t^{\mathcal{N}_i} \in \mathcal{I}_t^{\mathcal{N}_i}$ can only depend on information available at or before timestep $t$. Assuming an initial state $\mathbf{s}_0$, the expected cumulative reward of agent $i$ is:
\begin{multline}\label{eq:cum_reward}
    \tilde{R}^i(\mathbf{s}_0) = \mathbb{E} \Bigg [ \sum_{t \geq 0} \gamma^t R^i(\mathbf{s}_t, \mathbf{a}_t, \mathbf{s}_{t+1}) \mid \\
    a_t^i \sim \pi^i(\cdot \mid \mathbf{s}_t, I_t^{\mathcal{N}_i}), \mathbf{s}_0\Bigg ],
\end{multline}
where the subscript $t$ refers to the timestep and $\pi^i(\cdot \mid \cdot)$ denotes the policy of agent $i$.
The optimization problem that we aim to solve \revfin{in a distributed manner} is  formulated as follows:
\begin{equation}\label{eq:main_problem}
\begin{aligned}
\min_{\pi^i(\cdot \mid \cdot),\ \forall i \in \mathcal{N}} \quad & \max_{i \in \mathcal{N}}\ { -\tilde{R}^i(\mathbf{s}_0) }\\
\textrm{s.t.} \quad & \pi^i: \mathcal{S} \times \mathcal{I}_t^{\mathcal{N}_i} \rightarrow \Delta(\mathcal{A}^i),\ \forall i \in \mathcal{N}.\\
\end{aligned}
\end{equation}
We will call problem (\ref{eq:main_problem}) the main problem (min-max or max-min interchangeably) from now on. 
We can impose additional restrictions on the policies, such as requiring that they be deterministic.

\section{Motivating the Proposed Approach}
\revral{We highlight the challenges in solving the multi-agent min-max problem using online planning by showing how it differs from the single-agent case. With this analysis, our approach to solving the problem is motivated.}
\subsection{Online Planning Module}
This section highlights the differences between performing online planning in the single-agent min-max problem and the multi-agent min-max case. In the single-agent setting, the min-max criterion simply becomes the min criterion. \revral{The subsection on the single-agent case describes already known facts from the literature. The subsection on the multi-agent case motivates our novel algorithm by following a parallel structure to the previous subsection. It is an informal description of the reasoning behind our proposed algorithm.}
\subsubsection{\revral{The Single-Agent Case}}
Much of the success of online planning methods in single-agent reinforcement learning can be attributed to the Bellman equation \cite{kochenderfer2022algorithms}, which pertains to the underlying framework, i.e., the Markov decision process (MDP). An MDP is the degenerate case of problem (\ref{eq:main_problem}) for $\mathcal{N}=\lbrace 1 \rbrace$. According to the Bellman equation, the optimal solution to an MDP satisfies the principle of dynamic programming.  Also, if an optimal policy exists for an MDP, a deterministic optimal policy exists. Namely, an optimal action $a_\mathbf{s}$, when in state $\mathbf{s}$, in order to maximize the expected cumulative reward, is
\begin{multline} \label{eq: bellman_optimality}
    a_\mathbf{s} = \textrm{argmax}_{a \in \mathcal{A}^1}\ Q^*(\mathbf{s},a) = \\
    \textrm{argmax}_{a \in \mathcal{A}^1} \  R(\mathbf{s},a) + \gamma \sum_{\mathbf{s}' \in \mathcal{S}} P(\mathbf{s}' \mid \mathbf{s}, a) V^*(\mathbf{s}'),
\end{multline}
where $V^*(\mathbf{s}')$ denotes the optimal cumulative reward when starting at state $\mathbf{s}'$ and following an optimal policy for the MDP, $Q^*(\mathbf{s},a)$ denotes the return when taking action $a$ at $\mathbf{s}$ and then following an optimal policy, and $R(\mathbf{s},a)$ is the expected reward when transitioning from $\mathbf{s}$ with action $a$. 
The solution to an instance of the MDP is a sequence of actions (obtained via an optimal policy) that maximizes the expected cumulative reward. 

The online planning methods attempt to obtain an action sequence with a return realization close to the optimal one. They work as follows. 
At the first timestep, we assume that the state is $\mathbf{s}_0$. Then, \revfin{as a surrogate to eq. (\ref{eq: bellman_optimality})}, the online planning method creates a search tree with root node $\mathbf{s}_0$ and first-level branches pertaining to $L$ sampled actions $\{\bar{a}^{(1)}, \ldots, \bar{a}^{(L)}\}$ that are to be taken at $\mathbf{s}_0$. We assume that for these, the online planning method is able to obtain a good approximation of the value $Q^*(\mathbf{s}_0,\bar{a}^{(l)})$, which we denote $\tilde{Q}_{\mathbf{s}_0,\bar{a}^{(l)}}$, by simulating future action trajectories. Then, the algorithm chooses as $a_{\mathbf{s}_0}$ (the action that will actually be taken at $\mathbf{s}_0$) the action $\bar{a}^{(\textrm{argmax}_{1 \leq l \leq L}\ \tilde{Q}_{\mathbf{s}_0,\bar{a}^{(l)}})}$. The agent performs $a_{\mathbf{s}_0}$, transitioning to state $\mathbf{s}_1$ at the next timestep, and receiving reward $r_1$. At this time, the goal remains to optimize the cumulative reward since the first timestep. Hence, the agent should find $a_{\mathbf{s}_1}$ according to
\begin{equation}
    a_{\mathbf{s}_1} = \textrm{argmax}_{a \in \mathcal{A}^1}\ r_1 + \gamma {Q}^*(\mathbf{s}_1, a)= \textrm{argmax}_{a \in \mathcal{A}^1}\ {Q}^*(\mathbf{s}_1, a),
\end{equation}
and we see that we arrive at the same optimization problem we would get if we directly applied Bellman's equation at $\mathbf{s}_1$, \revfin{i.e., 
the past sequence of actions and rewards is irrelevant to the current action selection}.
\revfin{Hence, the same online planning procedure as above is followed for every timestep.} 

The main question is how can the online planning method obtain good estimates of $Q^*(\mathbf{s},a)$, which are denoted $\Tilde{Q}_{\mathbf{s},a}$, when $Q^*(\mathbf{s},a)$ internally assumes the agent follows the optimal policy. Online methods, such as MCTS, update the value or candidacy for optimality of an action $a$ at a given state $\mathbf{s}$ in the search tree using sampled trajectory paths and their cumulative rewards starting at that state-action pair. The value of $a$ at $\mathbf{s}$ is determined by the empirical average return of paths starting at this state-action pair, when at later states in the path, $\mathbf{s}'$, the action to be followed is chosen using the upper confidence bound (UCB1) exploration heuristic \cite{kochenderfer2022algorithms}. UCB1 uses the approximate $Q$ values, $\Tilde{Q}_{\mathbf{s}', \cdot}$, computed so far in the search tree for node $\mathbf{s}'$, and selects as next action the one that maximizes partly the $\Tilde{Q}_{\mathbf{s}', \cdot}$ value and partly an exploration bonus. Based on the received \revfin{returns, MCTS} updates the estimate $\Tilde{Q}_{\mathbf{s}, a}$. This indicates that the value of an action at state node $\mathbf{s}$ is determined with respect to a proxy to the best policy (since maximizing actions are taken in the path). Therefore, the online planning method can be seen as an approximate policy iteration algorithm \cite{kochenderfer2022algorithms}, where approximate $Q$ values are computed based on simulated trajectories (proxy of policy evaluation), and then actions are chosen in internal levels based on the UCB1 metric (which is a proxy for policy improvement), to finally find a good estimate of $Q^*$. 
Hence, single-agent online planning methods produce trajectory samples using action sequences that follow the basic rule of optimality for MDPs, eq. (\ref{eq: bellman_optimality}), which in turn follows the principle of dynamic programming. This explains why $\tilde{Q}_{\mathbf{s}, \bar{a}^{(l)}}$ is a good approximation of $Q^*(\mathbf{s}, a^{(l)})$ with the UCB1 metric.

\subsubsection{The Multi-Agent Case}
\label{sec:online_planning_multi}
In order to solve the main problem (\ref{eq:main_problem}) online, we follow the reasoning in the previous subsection. The goal is to maximize the worst agent's expected cumulative reward since the beginning of time. At the first timestep, the state is $\mathbf{s}_0$ and we use online planning to create the functions $\hat{Q}_0^i(\mathbf{s}_0, \mathbf{a}_0^{\mathcal{N}_i})$ at every agent $i$. These functions are to be good approximations of the local return functions at the optimal solution of problem (\ref{eq:main_problem}). They should express the local cumulative return, assuming that the agents follow the optimal policies of problem (\ref{eq:main_problem}) after timestep $0$ and at timestep $0$ take $\mathbf{a}_0$. Similarly, functions $\hat{Q}_t^i(\mathbf{s}_t, \mathbf{a}_t^{\mathcal{N}_i})$ should express the local cumulative return starting at $t$, assuming that the agents follow the optimal policies of problem (\ref{eq:main_problem}) before and after timestep $t$, and at timestep $t$ take $\mathbf{a}_t$. Then, the problem
\begin{equation}\label{eq:surrogate_1}
\max_{\mathbf{a}_0 \in \times_{i \in \mathcal{N}} \mathcal{A}^i} \quad \min_{i \in \mathcal{N}}\ {\hat{Q}_0^i(\mathbf{s}_0, \mathbf{a}_0^{\mathcal{N}_i}) }
\end{equation}
serves as an approximation of problem (\ref{eq:main_problem}) at the initial timestep. The agents solve problem (\ref{eq:surrogate_1}) in a distributed manner and take the action computed as the maximizer. Then each agent $i$ receives a local reward $r_1^i$ and the state becomes $\mathbf{s}_1$ at timestep 1. At this timestep, the agents will again apply online planning to continue solving the original problem (\ref{eq:main_problem}). They now create new function approximations $\gamma \hat{Q}_1^i(\mathbf{s}_1, \mathbf{a}_1^{\mathcal{N}_i})+r_1^i$, which should express the local cumulative rewards assuming that the agents follow the optimal policies of problem (\ref{eq:main_problem}) before and after timestep $1$ and at timestep $1$ take $\mathbf{a}_1$. The problem
\begin{equation}\label{eq:surrogate_2}
\max_{\mathbf{a}_1 \in \times_{i \in \mathcal{N}} \mathcal{A}^i} \quad \min_{i \in \mathcal{N}}\ {r_1^i+ \gamma \hat{Q}_1^i(\mathbf{s}_1, \mathbf{a}_1^{\mathcal{N}_i}) }
\end{equation}
serves as the new proxy of problem (\ref{eq:main_problem}) in timestep $1$ and is again solved using a distributed algorithm. The same process continues at every timestep. 

For this method to work, the online planning approach needs to obtain function approximations $\hat{Q}_t^i(\cdot, \cdot)$, which are good approximations of the returns of the agents when following the optimal policies of problem (\ref{eq:main_problem}). In our algorithm, we assume that, at every timestep $t$, each agent $i$ creates a search tree for its own local reward, \revfin{with root node $\mathbf{s}_t$}, which is used to compute $\hat{Q}_t^i(\mathbf{s}_t, \mathbf{a}_t^{\mathcal{N}_i})$. The function $\hat{Q}_t^i(\mathbf{s}_t, \mathbf{a}_t^{\mathcal{N}_i})$ is a concave function approximation of the values $\tilde{Q}_{t, \mathbf{s}_t, \bar{\mathbf{a}}^{\mathcal{N}_i, (l)}}^i$ for the sampled neighborhood action tuples $\bar{\mathbf{a}}^{\mathcal{N}_i, (l)}$ at the first level of the tree search of agent $i$ at timestep $t$, which has root node $\mathbf{s}_t$. In this case, using the UCB1 criterion for the expansion of the search tree of each agent need not be the approach that provides the best approximation of performance, $\hat{Q}_t^i(\mathbf{s}_t, \mathbf{a}_t^{\mathcal{N}_i})$, for the optimal policy of problem (\ref{eq:main_problem}). \revone{The reasons for this are twofold: the principle of dynamic programming does not necessarily hold for the solution of problem (\ref{eq:main_problem}), as shown in the following theorem, and the UCB1 in the search tree is performed only with respect to the local reward}. Nevertheless, UCB1 is still used and future work will investigate other criteria.

\begin{theorem}
\label{th:dyn_prog}
    For a problem of the form (\ref{eq:main_problem}), the principle of dynamic programming does not necessarily hold.
\end{theorem}
\begin{proof}
    The proof is based on a simple counter-example, where a sequence of actions that are non-optimal for their respective min-max subproblems, constitute an optimal action sequence for the min-max problem involving all timesteps. The details of the proof are included in the appendix.
\end{proof}

Overall, \revral{there are two significant challenges in applying online planning to solve the multi-agent min-max problem that do not appear in the single-agent case. }In the multi-agent min-max scenario it is important to keep in memory the cumulative reward received up to the current timestep for every agent. This is because the dynamic programming principle does not hold. In addition, while decisions can be made with $\Tilde{Q}_{\cdot, \cdot}$ values for sampled actions in the single agent setting, we require a function approximation $\hat{Q}_t^i(\cdot, \cdot)$ in the multi-agent scenario, \revral{because of the different trajectories sampled by each agent performing online planning.}

\subsection{Distributed Optimization Module}
As mentioned in the previous section, when solving the general version of problem (\ref{eq:main_problem}), each agent performs online planning to compute an approximation \revfin{of} its local cumulative reward and then the agents determine their current actions by solving the approximate max-min problem at the current timestep. This is done by deploying a distributed optimization framework. In our algorithm, we choose to use the method presented by Srivastava et al. \cite{distr_opt}. This is a distributed stochastic subgradient method. 
We decide to use a subgradient method, because it is generally applicable: 
any convex function has a nonempty and bounded subdifferential in the interior of its domain. The distributed optimization framework by Srivastava et al. \cite{distr_opt} solves the problem
\begin{equation}\label{eq:distr_opt}
\min_{\mathbf{a} \in \times_{i \in \mathcal{N}} \mathcal{A}^i} \quad \max_{i \in \mathcal{N}}\ { f^i(\mathbf{a}}),
\end{equation}
where each $f^i$ is convex and known only to agent $i$. The closed and convex set $\times_{i \in \mathcal{N}} \mathcal{A}^i$ ($\mathbf{a} \in \times_{i \in \mathcal{N}} \mathcal{A}^i$) is known to all agents. We now look at the method from the perspective of agent $i$. The estimate of the optimal point and optimal value of problem (\ref{eq:distr_opt}) of agent $i$, at iteration $k$, are denoted $\boldsymbol{\alpha}_k^i$ and $\eta_k^i$, respectively. The method proceeds with two steps at every iteration $k \geq 0$. At first, it forms an intermittent adjustment, and then agent $i$ takes a step towards minimizing its own function $f^i$. This is written:
\begin{equation}
    \begin{bmatrix}
        \tilde{\boldsymbol{\alpha}}^i_k \\
        \tilde{\eta}^i_k
    \end{bmatrix}=\sum_{j \in \mathcal{N}_i} w^{ij}_k  \begin{bmatrix}
\boldsymbol{\alpha}_k^i\\
\eta_k^i
\end{bmatrix},\     \mathbf{v}_k^i =\begin{bmatrix}
        \tilde{\boldsymbol{\alpha}}^i_k \\
        \tilde{\eta}^i_k
    \end{bmatrix} - \dfrac{\beta_k}{N} \begin{bmatrix}
\mathbf{0}_n\\
1
\end{bmatrix},
\end{equation}
\begin{equation}\label{eq:distr_opt_end}
       \begin{bmatrix}
{\boldsymbol{\alpha}}_{k+1}^i\\
{\eta}_{k+1}^i
\end{bmatrix} = \Pi_{\mathcal{A} \times \mathbb{R}} \left[ \mathbf{v}_k^i - \beta_kr^i \left(g^i + \begin{bmatrix}
\epsilon_k^i\\
0
\end{bmatrix} \right)\right].
\end{equation}
In the equations above, $w^{ij}_k$ denotes the
weight that agent $i$ assigns to its neighbor $j$ at timestep $k$, $\beta_k$ is the step 
size, $\epsilon_k^i$ the random subgradient error, $r^i > 1$, $g^i$ a subgradient of $\tilde{g}^i
(\boldsymbol{\alpha}, \eta) = \tilde{g}^i
(\mathbf{z}) = \mathrm{max}\lbrace0, f^i(\boldsymbol{\alpha})-\eta\rbrace$ at $\mathbf{v}_k^i$, and $n$ denotes the dimensionality of $\mathbf{a} \in \times_{i \in \mathcal{N}} \mathcal{A}^i$. Here, $\Pi_{\mathcal{A}\times \mathbb{R}}$ denotes the Euclidean projection on the set $\mathcal{A}\times \mathbb{R}$.

Under noise, subgradient, communication topology, and assigned weights assumptions, it is guaranteed that the iterates $\mathbf{z}_k^i = (\boldsymbol{\alpha}_k^i,\ \eta_k^i)$ converge to a common $\mathbf{z}^*=(\boldsymbol{\alpha}^*,\ \eta^*)$ such that $\boldsymbol{\alpha}^*,\ \eta^*$ are an optimal point and the optimal value of problem (\ref{eq:distr_opt}), respectively \cite{distr_opt}.

\section{Proposed Algorithm}
Our \revral{proposed} algorithm is Algorithm \ref{alg:proposed}. At every planning step, it passes through two phases. In the first phase, each agent $i$ creates, through online planning, \revral{a convex} function approximation of its local expected cumulative \revral{cost}, which is a function of only the immediate next action of the agents in its neighborhood. \revone{In the current instance of the algorithm, \revral{the online planning module uses the POMCPOW algorithm \cite{sunberg2018online} (with the observation equaling the state)}. The convex function approximator used is the adaptive max-affine partitioning algorithm (AMAP) \cite{balazs2016convex}, which offers a balance between the regression speed and model quality.} The steps mentioned above correspond to lines 5--8 in Algorithm \ref{alg:proposed}. \revral{In the second phase, the subgradient-based distributed optimization algorithm \cite{distr_opt} is deployed to solve the min-max problem involving the convex function approximations from the first phase and determine the agents' next action (lines 9--12).} In Algorithm \ref{alg:proposed}, we give our proposed algorithm from the perspective of agent $i$. Because of properties (\ref{eq:reward}-\ref{eq:P_struct}), agent $i$ only needs to create a search tree with actions of its neighbors and the portion of the state relative to its neighborhood.

\algdef{SE}[SUBALG]{Indent}{EndIndent}{}{\algorithmicend\ }%
\algtext*{Indent}
\algtext*{EndIndent}
\begin{algorithm}
\caption{Proposed algorithm from the perspective of agent $i$ at timestep $t$}
\label{alg:proposed}
	\begin{algorithmic}[1]
        \State \textbf{Inputs:}
            \Indent 
            \State planning timestep $t$
            \State current environment state $\mathbf{s}_t$
            \State cumulative reward up to timestep $t-1$: $\bar{R}^i_{t-1}$
            \EndIndent
            \State obtain $L$ samples of future expected cumulative reward from state $\mathbf{s}_t^{\mathcal{N}_i}$: using online planning (e.g., MCTS), sample an action tuple $\mathbf{a}^{\mathcal{N}_i, (l)}$ and compute $\Tilde{Q}^i_{t, \mathbf{s}^{\mathcal{N}_i}_t, \mathbf{a}^{\mathcal{N}_i, (l)}},\ l \in 1, \dots, L$
            \State negate the sampled cumulative rewards to get future cumulative cost samples $-\Tilde{Q}^i_{t, \mathbf{s}^{\mathcal{N}_i}_t, \mathbf{a}^{\mathcal{N}_i, (l)}}$
            \State create a convex function approximation $-\hat{Q}^i_t(\mathbf{s}_t^{\mathcal{N}_i}, \mathbf{a}^{\mathcal{N}_i})$ based on the samples $-\Tilde{Q}^i_{t, \mathbf{s}^{\mathcal{N}_i}_t, \mathbf{a}^{\mathcal{N}_i, (l)}}$
            \State add $-\bar{R}^i_{t-1}$ to $-\gamma^t \hat{Q}^i_t(\mathbf{s}_t^{\mathcal{N}_i}, \mathbf{a}^{\mathcal{N}_i})$ and get $f^i_t(\mathbf{s}_t^{\mathcal{N}_i}, \mathbf{a}^{\mathcal{N}_i})$

            \State initialize the estimates $\boldsymbol{\alpha}^i_0 \in \mathcal{A}$, $\eta_0^i \in \mathbb{R}$ 
            \For {$k = 0, \dots, K$}
                
                \State $\boldsymbol{\alpha}^i_{k+1}, \eta_{k+1}^i$ $\leftarrow$ one step of algorithm (\ref{eq:distr_opt_end})
            \EndFor
            \State perform the action pertaining to agent $i$ in $\boldsymbol{\alpha}_{K+1}^i$
            \State collect immediate reward $r_t^i$ and $\bar{R}_{t}^i \leftarrow \bar{R}_{t-1}^i +\gamma^t r_t^i$
\end{algorithmic}
\end{algorithm}

\revral{The number of samples $L$ controls the accuracy of the cost function approximation. The number of iterations $K$ determines the consensus among the agent iterates and their distance to the minimum of the optimization problem. The agent iterates determine the next action for the agents. }

\subsection{Computational Complexity}
\revral{
The per-agent computational complexity of the proposed algorithm is completely determined by the three submodules included in Algorithm \ref{alg:proposed}: the online planning module (POMCPOW) in line 5, the convex function approximator \cite{balazs2016convex} in line 7, and the distributed optimization framework in lines 10--11. In Algorithm \ref{alg:proposed}, POMCPOW uses $L$ planning calls. With a max tree depth $\Upsilon$, it has complexity $\mathcal{O}(L \Upsilon \log L)$ \cite{sunberg2018online}. The AMAP model improvement step has complexity $\mathcal{O}(\max \lbrace \phi_{\mathrm{LSPA}}, D_i\rbrace \max \lbrace H, D_i\rbrace D_i L)$, when $H$ hyperplanes are used, $D_i$ is the argument dimensionality of $-\hat{Q}^i_t(\mathbf{s}_t^{\mathcal{N}_i}, \cdot)$, and $\phi_{\mathrm{LSPA}}$ denotes the number of iterations of a sub-algorithm \cite{balazs2016convex}. Assuming an ensemble of $M$ models and multiple model improvements, the AMAP total complexity is $\mathcal{O}(L^{D_i/D_i+4}M\max \lbrace \phi_{\mathrm{LSPA}}, D_i\rbrace \max \lbrace H, D_i\rbrace D_i L + M L D_i^2)$. The second term comes from the initial least-squares problems.  Finally $K$ iterations of the distributed optimization algorithm are performed, with the complexity of each dominated by the complexity of eq. (12). The complexity of eq. (12) can greatly vary depending on the structure of $\mathcal{A}$, but certain cases allow for analytic solutions \cite{boyd2004convex}. Therefore we suppose a generic complexity for eq. (12), $\mathcal{O}(\Pi)$. In total, our algorithm's complexity is the sum of the above complexities. We however note that the method used for each of the submodules of Algorithm \ref{alg:proposed} can be independently replaced, e.g., online planning can be performed with an algorithm other than POMCPOW.

}

\section{Experiments}
\revone{\revral{In order to evaluate the performance of our algorithm, we pose the formation control problem as a min-max problem. In this setting, the agents aim to converge to states such that desired relative states are satisfied between them} \cite{mesbahi2010graph}.\footnote{The source code for the experiments can be found at: \url{https://github.com/alextzik/distr_online_maxmin_markov_game}.}}
\revone{\subsection{Experiment Setup}}
We consider a Markov game where the state consists of the $x, y$ positions of the agents, i.e., $\mathbf{s} = (s^1, \dots, s^N) \in \mathcal{S} = \times_{i \in \mathcal{N}} \mathbb{R}^2 $, where $s^i = (x^i, y^i)$ the position of agent $i$. The dynamics of each agent $i$, from timestep $t$ to $t+1$, with action $a_{t}^i$, evolve according to
\begin{equation}
s_{t+1}^i = 
    s_{t}^i + 
    \begin{bmatrix}
    1 & 0\\
    -1 & 2
    \end{bmatrix}a_{t}^i,\ a_t^i = (a, b),\ a, b \in \left[ 0, 1\right].
\end{equation}
We assume deterministic dynamics, in order to allow comparisons with various techniques.
The instantaneous reward for agent $i$ is 
\begin{equation}
    R^i(\mathbf{s}^{\mathcal{N}_i}, \mathbf{a}^{\mathcal{N}_i}, 
    \mathbf{s'}^{\mathcal{N}_i}) = \sum_{j \neq l \in \mathcal{N}_i} \lVert (s^j-s^l) - (d^j-d^l)\rVert,
\end{equation}
where $d^j$ denotes the desired position of agent $j$ in a specific formation that satisfies the relative states. Our goal is to  solve problem (\ref{eq:main_problem}), which also leads the agents to satisfy their desired relative states. A reward close to $0$ signifies satisfaction of the desired relative states. We assume $\gamma=1$ and a finite horizon of $T=150$.

\revral{We consider three graph topologies $\mathcal{G}_1, \mathcal{G}_2$, and $\mathcal{G}_3$. In $\mathcal{G}_1$, five agent communicate over an almost fully connected graph (only agent 1 is not a neighbor of agent 5). In $\mathcal{G}_2$, five agent have the topology 1-2-3-4-5-1. In $\mathcal{G}_3$, eight agents communicate over the topology 1-2-3-4-5-6-7-8.}

\subsection{Implementation Details for the Proposed Algorithm}
Our proposed approach has the following parameters in Algorithm \ref{alg:proposed}: $L=100$, and $K=1000$. For the online framework we deploy POMCPOW \cite{sunberg2018online}, since both our action and state spaces are continuous. \revone{Every agent $i$ has access to the positions of its neighbors, $\mathbf{s}^{\mathcal{N}_i}$}. Internally in POMCPOW of agent $i$, the observation equals the next simulated state for the agents in $\mathcal{N}_i$. For the POMCPOW implementation at each agent, we use the action and observation progressive widening values of $k_a=2.0$, $\alpha_a = 0.5$,  $k_o=0.0$, and $\alpha_o = 0.5$ in Listing (1) of \cite{sunberg2018online}. It also holds that the $n$ of Listing (1) in \cite{sunberg2018online} equals $L$ in our case. The search depth is $5$. The rollout planner chosen for our algorithm is based on the affine formation control system
\begin{equation}\label{eq:linear_formation_control}
    \dfrac{ds^i(t)}{dt} = -\sum_{j \in \mathcal{N}_i} \left( s^i(t) - s^j(t) \right) - \left(d^i-d^j \right),\ \forall i,
\end{equation}
for which it is guaranteed that $s^i(t) \rightarrow d^i+\xi$ as $t \rightarrow \infty$, where $\xi$ is a constant displacement vector \cite{mesbahi2010graph}. Therefore, if the rollout policy is called at a given depth in the search tree of an agent, it first solves 
the system of differential equations (\ref{eq:linear_formation_control}) with the state of the current node in the search tree as the initial condition and then chooses for the agents the action that takes them closer to the converged states of eq. (\ref{eq:linear_formation_control}), in the sense of a one-step look-ahead, which is an easily solvable convex problem. Note that the rollout policy is used within the POMCPOW of each agent $r$. Hence, in POMCPOW of agent $r$, eq. (\ref{eq:linear_formation_control}) only involves the agents $i$ such that $i \in \mathcal{N}_r$, and agent $r$ can assume a fully connected topology in the rollout planner. After simulating the actions for the agents, the search moves to a new node, if the maximum depth has not been reached. Note that different initializations of eq. (\ref{eq:linear_formation_control}) lead to different converged positions, which means that simply using the rollout planner should not solve our problem. In addition, the converged state for agent $i$ can be different in the individual systems solved by the agents $j \in \mathcal{N}_i$.

\revone{\subsection{Baseline Algorithm Description}}
Using the rollout planner at every timestep with depth $1$ constitutes a simple baseline. The agents solve eq. (\ref{eq:linear_formation_control}), for the true graph topology and involving every $i \in \mathcal{N}$, in a distributed manner and then each agent takes the action that in the one step look-ahead notion brings it closer to its converged state for eq. (\ref{eq:linear_formation_control}). This is the `rollout baseline' below. \revral{Assuming Euler discretization with step $\Delta t$ and total time $T_f$, each timestep of the rollout baseline at agent $i$ is $\mathcal{O}(\frac{T_f}{\Delta t}\lvert \mathcal{N}_i\rvert)$.}
The second baseline is the `POMCPOW baseline', where each agent uses Algorithm \ref{alg:proposed} up to line 5, \revral{with complexity $\mathcal{O}(L \Upsilon \log L)$}. It then follows the action pertaining to itself from the sample, $\Tilde{Q}^i_{t, \mathbf{s}^{\mathcal{N}_i}_t, \mathbf{a}^{\mathcal{N}_i, (l)}}$ with the maximum value. \revral{Obviously, the proposed algorithm is the most computationally intensive among the compared methods, but with significantly better min-max behavior, as shown below.}

In the scenario that we are considering, we can obtain an upper bound on performance, by finding the optimal point and value of problem (\ref{eq:main_problem}).
Since the proposed problem is deterministic, the following centralized 
convex problem is solved by the optimal open-loop sequence of actions for the agents 
\begin{equation}\label{eq:cvx_highbound}
    \min_{\mathbf{a}^i_0, \dots, \mathbf{a}^i_{T-1},\ \mathbf{s}^i_0, \dots,
    \mathbf{s}^i_{T}\forall\ i \in \mathcal{N}} 
    \quad  \max_{i \in \mathcal{N}}\ 
    {-\sum_{t=0}^{T-1} 
    {R}^i(\mathbf{s}^{\mathcal{N}_i}_t, 
    \mathbf{a}^{\mathcal{N}_i}_t, 
    \mathbf{s}^{\mathcal{N}_i}_{t+1}) }.
\end{equation}
Problem (\ref{eq:cvx_highbound}) serves as the main comparison to our work. It is referred to `optimal' below. We have chosen this simplified application in order to make the upper bound of performance tractable.

\revone{\subsection{Performance Evaluation}}
We present results for the instantaneous reward of the worst-performing agent \revral{(in terms of return, which is the max-min objective)} as a function of planning step in \revral{Figures \ref{fig:reward_5_agent}--\ref{fig:reward_8_agent}. In Figure \ref{fig:reward_5_agent}, five agents communicate over $\mathcal{G}_1$ for all timesteps. In Figure \ref{fig:reward_5_agent_switching}, five agents communicate over a switching network topology: every 10 timesteps, the topology switches between $\mathcal{G}_1$ and $\mathcal{G}_2$. In Figure \ref{fig:reward_8_agent}, eight agents communicate over the fixed topology $\mathcal{G}_3$.}

Our proposed algorithm performs better than the baselines by a wide margin. The POMCPOW baseline cannot achieve optimal performance, because it does not lead to coordination as explained above. The rollout baseline is also not able to converge to the desired configuration, because the converged states keep changing at the solution of eq. (\ref{eq:linear_formation_control}) at every timestep. 
\revral{The oscillatory behavior of reward for the two baselines in Figure \ref{fig:reward_5_agent_switching} is attributed to the changing topology, which alters the number of terms in the reward sum. Nevertheless, the reward within the timesteps of constant topology is not increasing, which indicates a lack of coordination for the reasons explained above.}
Our algorithm is able to reach an $\epsilon$-suboptimal configuration, but cannot exactly reach the optimal formation, because \revral{the agents' function approximations, used in the min-max optimization problem to obtain the next actions, are not minimized by the zero action at later steps}. \revral{This explains the observed oscillatory performance. We expect smaller oscillations as $L$ and tree density in POMCPOW and $H$ in AMAP increase. Our proposed algorithm is able to obtain convergence rates similar to the open-loop optimal action sequence, constrained however by the connectedness of the underlying communication topology.}

\begin{figure}
     \centering
    \includegraphics[width=0.47\textwidth]{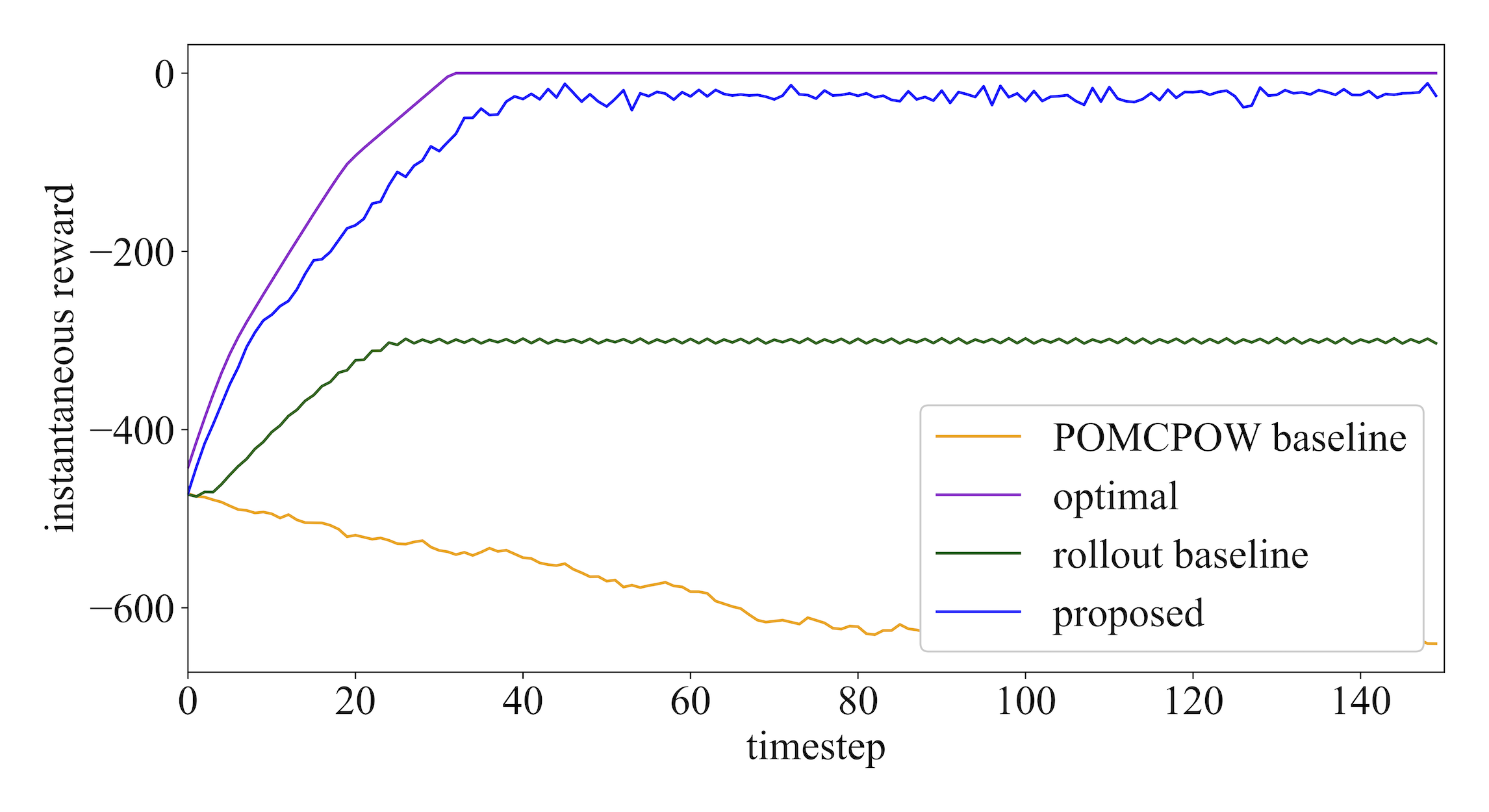}
        \caption{Instantaneous reward of the worst-performing agent for $\mathcal{G}_1$. \revral{Five} agents communicate over an almost fully connected communication network (only agent \revral{1} is not a neighbor of agent \revral{5}). We observe that our proposed method performs better than the baselines on the max-min criterion.}
        \label{fig:reward_5_agent}
\end{figure}

\begin{figure}
     \centering
    \includegraphics[width=0.45\textwidth]{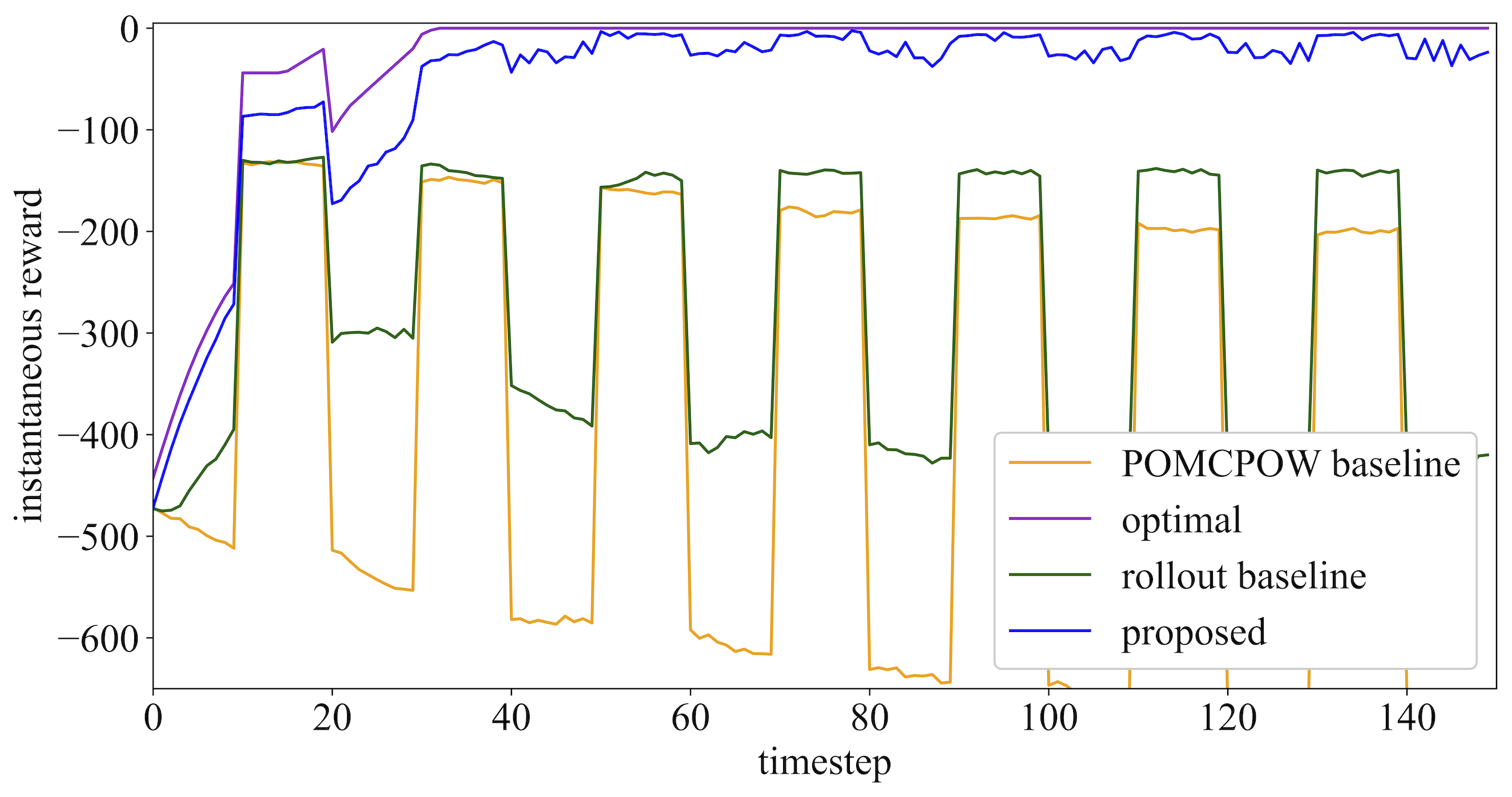}
        \caption{\revral{Instantaneous reward of the worst-performing agent for a switching topology of five agents. Every $10$ timesteps the topology changes between $\mathcal{G}_1$, an almost fully connected topology, and $\mathcal{G}_2$, which is a cyclic graph. We observe that our proposed method performs better than the baselines on the max-min criterion.}}
        \label{fig:reward_5_agent_switching}
\end{figure}

\begin{figure}
     \centering
         \centering
    \includegraphics[width=0.47\textwidth]{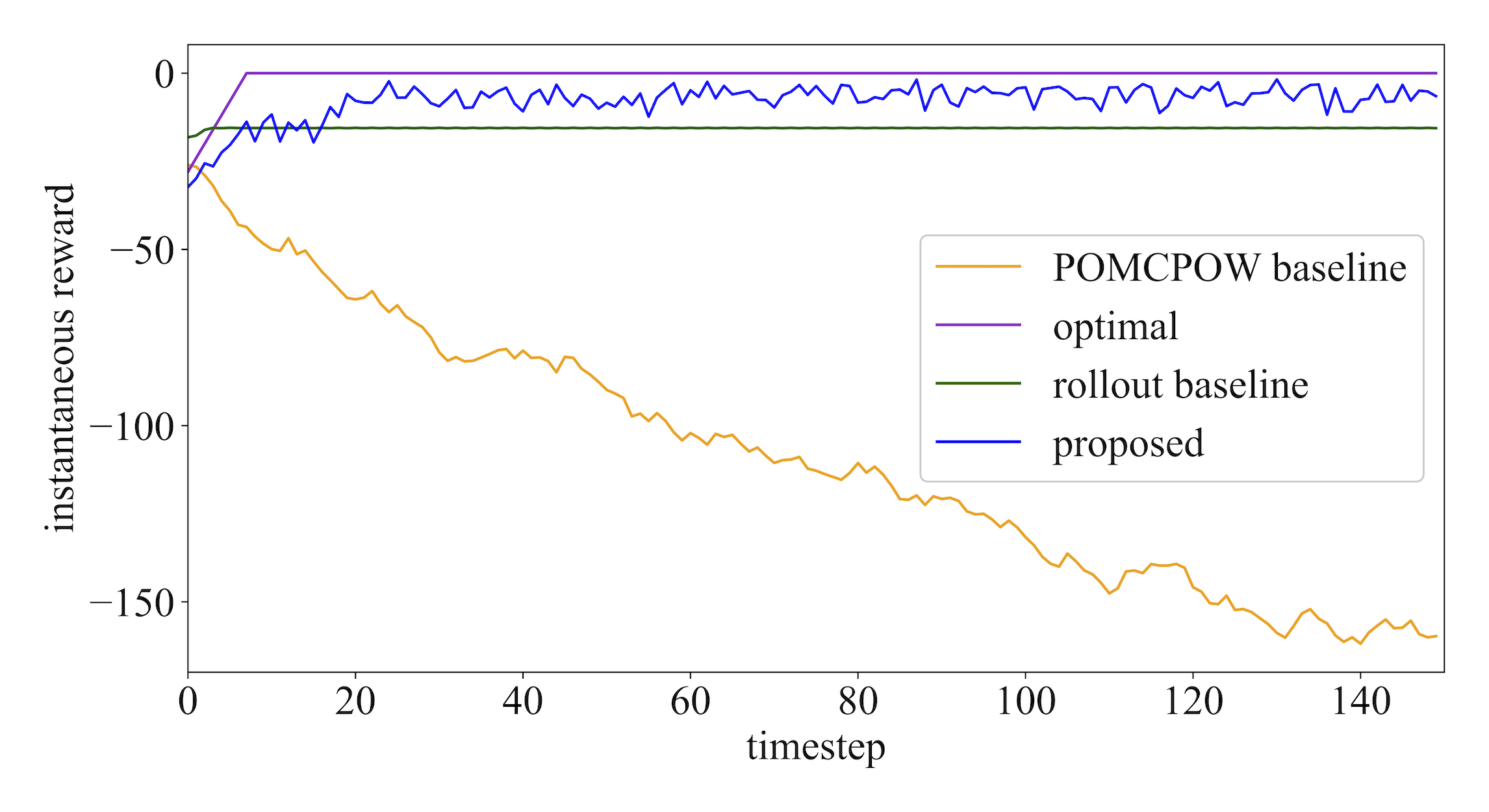}
     \caption{\revzero{Instantaneous reward of the worst-performing agent for $\mathcal{G}_3$. \revral{Eight} agents communicate over the topology \revral{1-2-3-4-5-6-7-8}. The worst agent cumulative reward in our method is \revral{-1188}, while the worst agent return in the POMCPOW baseline is \revral{-2338}. We observe that our proposed method performs better than the baselines with respect to the max-min criterion.}}
     \label{fig:reward_8_agent}
\end{figure}

\vspace{0.3cm}
\section{Limitations and Future Work}
\revral{Our approach is superior to the baselines presented, but has limitations. We observe that the online planning and function approximator modules of our algorithm lead to oscillatory behavior. Tuning the parameters and exploring ways to mitigate this would be important before deployment in actual robots. Because computation is limited in robotic applications, an efficient implementation of the submodules of Algorithm \ref{alg:proposed} is also needed. In the current implementation, the online planning framework is the bottleneck in performance.}

From a theoretical perspective, as discussed in Section \ref{sec:online_planning_multi}, the UCB1 metric need not be the best exploration metric for the min-max problem. Theoretical analysis of the reasons behind its sub-optimality, and determination of a better metric are potential avenues of research. 
Our algorithm is also only approximate. Future work will aim towards a distributed algorithm that can provably converge to the optimal policies for the agents in the context of problem (\ref{eq:main_problem}). Finally, we assume that every agent knows the complete feasible set $\mathcal{A}$ for the actions of all agents beforehand.

\section*{Appendix: Proof of Theorem \ref{th:dyn_prog}}

    Assume the case of $2$ neighboring agents interacting over the course of $2$ timesteps (final state at $t=3$). Also assume $\mathcal{I}_t^{\mathcal{N}_i}$ is the empty set and the dynamics are deterministic: $s_{t+1} = s_{t} + a_t^1 + a_t^2$. The initial state is $s_1 = 0$, $a_t^i \in \lbrace 1, 2 \rbrace$, and $\gamma = 1$. The allowed policies are deterministic. 

    At timestep $2$ we can only have $s_2 \in \lbrace 2, 3, 4 \rbrace$. If:
    \begin{align}
        \min_{i \in \lbrace 1, 2\rbrace} &R^i(2, (1,1)) = \min \lbrace 90, 100  \rbrace  >\\
        &\min_{i \in \lbrace 1, 2\rbrace} R^i(2, \mathbf{a}),\ \forall \mathbf{a} \neq (1,1),\\
        &R^1(2, \mathbf{a}) = 200,  R^2(2, \mathbf{a}) = 50, \forall \mathbf{a} \neq (1,1),
    \end{align}
    then the optimal (with respect to the max-min criterion) policy for each agent in the subproblem starting at state $2$ is action $1$. Similarly, assume:
    \begin{align}
\min_{i \in \lbrace 1, 2\rbrace} &R^i(3, (1,2)) = \min \lbrace 90, 100 \rbrace >\\
&\min_{i \in \lbrace 1, 2\rbrace} R^i(3, \mathbf{a}),\ \forall \mathbf{a} \neq (1,2),\\
&R^1(3, \mathbf{a}) = 200,  R^2(3, \mathbf{a}) = 50, \forall \mathbf{a} \neq (1,2)
    \end{align}
    and
    \begin{align}
\min_{i \in \lbrace 1, 2\rbrace} &R^i(4, (2,1)) = \min \lbrace 90, 100 \rbrace >\\
&\min_{i \in \lbrace 1, 2\rbrace} R^i(4, \mathbf{a}),\ \forall \mathbf{a} \neq (2, 1),\\
&R^1(4, \mathbf{a}) = 200,  R^2(4, \mathbf{a}) = 50, \forall \mathbf{a} \neq (2,1).
    \end{align}
    Now, assume $R^1(0, \mathbf{a}) = 5$, $R^2(0, \mathbf{a}) = 200$ for all $\mathbf{a}$. Then, to solve the min-max problem (\ref{eq:main_problem}), the agents should choose any action when in state $s_1=0$ and for the state $s_2$ the system arrives at, they should choose any action other than the one optimal for the min-max subproblem starting at that state. In other words, suboptimal policies for the subproblems that could arise in timestep $2$ have better (greater) overall max-min objective value in the complete problem than the optimal policies for these subproblems. This completes the proof.

\section*{Acknowledgments}
This material is based upon work supported by the Under Secretary of Defense for Research and Engineering under Air Force Contract No. FA8702-15-D-0001. Any opinions, findings, conclusions or recommendations expressed in this material are those of the authors and do not necessarily reflect the views of the Under Secretary of Defense for Research and Engineering. 

\revral{Toyota Research Institute (TRI) also provided funds to assist the authors with this research, but this article solely reflects the opinions and conclusions of its authors and not TRI or any other Toyota entity. }

For the first author, this work is also
partially funded through the Alexander S. Onassis Foundation Scholarship program.

\bibliographystyle{IEEEtran}
\bibliography{IEEEabrv, sample}

\end{document}